\documentclass{llncs}
\usepackage{amsmath,amssymb}
\usepackage[dvipdfmx]{graphicx}
\usepackage[dvipdfmx]{color}

\newcommand{\msize}[1]{{\left|#1\right|}}
\newcommand{\Nei}[2]{N(#1,#2)}
\newcommand{\Neiclosed}[2]{N[#1, #2]}
\newcommand{\Tsub}[2]{T_{#1}^{#2}}

\newcommand{\bfI}{I}

\newcommand{\calS}{{\cal S}}

\newcommand{\onestep}{\leftrightarrow}
\newcommand{\sevstep}{\leftrightsquigarrow}

\newcommand{\sevstepT}[1]{\overset{#1}{\sevstep}}

\newcommand{\ImSet}[1]{{\sf R}(\bfI_{#1})}
\newcommand{\ImSetp}{{\sf R}(\bfI^\prime)}
\newcommand{\Forest}{F}
\newcommand{\numq}{q}

\newcommand{\bfIstar}{\bfI^*}
\newcommand{\bfIp}{\bfI^{\prime}}

\newcommand{\dist}{\mathsf{dist}}

\newcommand{\Tminus}{\bar{T}}
\newcommand{\Tmsub}[2]{\Tminus_{#1}^{#2}}
\newcommand{\Iminus}{\bar{I}}

\newenvironment{listing}[1]{%
        \begin{list}{*}{%
                 \settowidth{\labelwidth}{#1}%
                 \setlength{\leftmargin}{\labelwidth}%
                 \advance \leftmargin by 12pt
                   \setlength{\itemsep}{0pt}%
                   \setlength{\parsep}{0pt}%
                   \setlength{\topsep}{0pt}%
                   \setlength{\parskip}{0pt}%
}%
}{%
\end{list}}

\newcounter{one}
\newcommand{\one}{{\rm \roman{one}}}
\newcounter{two}
\newcommand{\two}{{\rm \roman{two}}}
\newcounter{three}
\newcommand{\three}{{\rm \roman{three}}}
\newcounter{four}
\newcounter{five}
\newcounter{six}
\definecolor{lightblue}{rgb}{0.5,0.5,1.0}
\definecolor{darkred}{rgb}{0.5,0,0}
\definecolor{darkgreen}{rgb}{0,0.5,0}
\definecolor{darkblue}{rgb}{0,0,0.5}

\begin{document}
\title{Linear-Time Algorithm for Sliding Tokens on Trees}

\author{
Erik D.~Demaine\inst{1} \and
Martin L.~Demaine\inst{1} \and 
Eli Fox-Epstein\inst{2} \and 
Duc A.~Hoang\inst{3} \and \\
Takehiro Ito\inst{4} \and 
Hirotaka Ono\inst{5} \and
Yota Otachi\inst{3} \and 
Ryuhei Uehara\inst{3} \and 
Takeshi Yamada\inst{3}
}

\institute{
    MIT Computer Science and Artificial Intelligence Laboratory,\\
    32 Vassar St., Cambridge, MA 02139, USA.\\
	\email{\{edemaine,~mdemaine\}@mit.edu}
\and
	Department of Computer Science, Brown University, \\
	115 Waterman Street, Providence, RI 02912-1910, USA. \\
	\email{ef@cs.brown.edu}
\and
    School of Information Science, JAIST, \\
    Asahidai 1-1, Nomi, Ishikawa 923-1292, Japan.\\
    \email{\{hoanganhduc, otachi, uehara, tyama\}@jaist.ac.jp}
\and
	Graduate School of Information Sciences, Tohoku University, \\
    Aoba-yama 6-6-05, Sendai, 980-8579, Japan.\\
	\email{takehiro@ecei.tohoku.ac.jp}
\and	
	Faculty of Economics, Kyushu University, \\
	Hakozaki 6-19-1, Higashi-ku, Fukuoka, 812-8581, Japan. \\
	\email{hirotaka@econ.kyushu-u.ac.jp}
}

\maketitle

\begin{abstract}
Suppose that we are given two independent sets $\bfI_b$ and $\bfI_r$ of 
a graph such that $\msize{\bfI_b}=\msize{\bfI_r}$, 
and imagine that a token is placed on each vertex in $\bfI_b$. 
Then, the {\sc sliding token} problem is to determine whether 
there exists a sequence of independent sets which transforms $\bfI_b$ 
into $\bfI_r$ so that each independent set in the sequence results from 
the previous one by sliding exactly one token along an edge in the graph. 
This problem is known to be PSPACE-complete even for planar graphs, and also for bounded treewidth graphs. 
In this paper, we thus study the problem restricted to trees, and give the following three results:
(1)~the decision problem is solvable in linear time; 
(2)~for a yes-instance, we can find in quadratic time an actual sequence of independent sets between $\bfI_b$ and $\bfI_r$ whose length (i.e., the number of token-slides) is quadratic; and 
(3)~there exists an infinite family of instances on paths for which any sequence requires quadratic length.
\end{abstract}

\section{Introduction}

Recently, {\em reconfiguration problems} attract the attention 
in the field of theoretical computer science. 
The problem arises when we wish to find a step-by-step transformation between 
two feasible solutions of a problem such that 
all intermediate results are also feasible and 
each step abides by a fixed reconfiguration rule 
(i.e., an adjacency relation defined on feasible solutions of the original problem).
This kind of reconfiguration problem has been studied extensively 
for several well-known problems, including 
{\sc independent set}~\cite{BB14,Bon14,BKW14,HearnDemaine2005,HearnDemaine2009,IDHPSUU,ItoKaminskiOnoSuzukiUeharaYamanaka2014,KaminskiMedvedevMilanic2012,MNRSS13,MNRW14,Wro14}, 
{\sc satisfiability}~\cite{Kolaitis,MTY11}, 
{\sc set cover}, {\sc clique}, {\sc matching}~\cite{IDHPSUU}, 
{\sc vertex-coloring}~\cite{BJLPP14,BC09,CHJ11,Wro14},
{\sc list edge-coloring}~\cite{IKD09,IKZ11},
{\sc list $L(2,1)$-labeling}~\cite{IKOZ_isaac},
{\sc subset sum}~\cite{ID11}, 
{\sc shortest path}~\cite{Bon13,KMP11}, and so on.

\subsection{{\sc Sliding token}}
The {\sc sliding token} problem was introduced by Hearn and Demaine~\cite{HearnDemaine2005} 
as a one-player game, which can be seen as a reconfiguration problem for {\sc independent set}. 
Recall that an {\em independent set} of a graph $G$ is a vertex-subset of $G$ in which 
no two vertices are adjacent. 
(Figure~\ref{fig:example} depicts five different independent sets in the same graph.)
Suppose that we are given two independent sets $\bfI_b$ and $\bfI_r$ of a graph 
$G = (V,E)$ such that $\msize{\bfI_b}=\msize{\bfI_r}$, 
and imagine that a token (coin) is placed on each vertex in $\bfI_b$. 
Then, the {\sc sliding token} problem is to determine 
whether there exists a sequence 
$\langle \bfI_1, \bfI_2, \ldots, \bfI_{\ell} \rangle$ of independent sets of $G$ such that
\begin{listing}{aaa}
\item[(a)] $\bfI_1=\bfI_b$, $\bfI_{\ell}=\bfI_r$, 
 and $\msize{\bfI_i} = \msize{\bfI_b}=\msize{\bfI_r}$ for all $i$, $1 \le i \le \ell$; and 
\item[(b)] for each $i$, $2 \le i \le \ell$, 
 there is an edge $\{u,v\}$ in $G$ such that $\bfI_{i-1} \setminus\bfI_{i}=\{u\}$ 
 and $\bfI_{i}\setminus\bfI_{i-1}=\{v\}$, 
 that is, $\bfI_{i}$ can be obtained from $\bfI_{i-1}$ by sliding exactly 
 one token on a vertex $u \in \bfI_{i-1}$ to its adjacent vertex $v$ along $\{u,v\} \in E$.
\end{listing}
Such a sequence is called a {\em reconfiguration sequence} between $\bfI_b$ and $\bfI_r$. 
Figure~\ref{fig:example} illustrates a reconfiguration sequence 
$\langle \bfI_1, \bfI_2, \ldots, \bfI_5 \rangle$ of independent sets 
which transforms $\bfI_b = \bfI_1$ into $\bfI_r = \bfI_5$. 
Hearn and Demaine proved that {\sc sliding token} is PSPACE-complete for planar graphs, 
as an example of the application of their powerful tool, 
called the nondeterministic constraint logic model, 
which can be used to prove PSPACE-hardness of 
many puzzles and games~\cite{HearnDemaine2005}, \cite[Sec.~9.5]{HearnDemaine2009}. 

\begin{figure}[t]
\begin{center}
\includegraphics[width=0.9\textwidth]{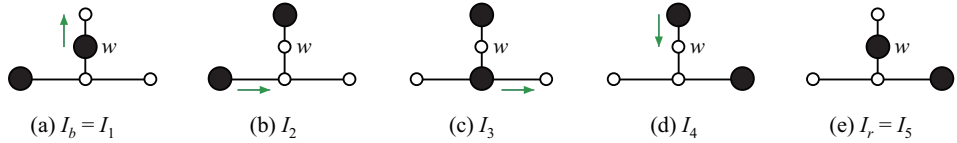}
\end{center}
\vspace{-1em}
\caption{A sequence $\langle \bfI_1, \bfI_2, \ldots, \bfI_5 \rangle$ of independent sets of the same graph, where the vertices in independent sets are depicted by large black circles (tokens).}
\label{fig:example}
\end{figure}

\begin{figure}[b]
\begin{center}
\includegraphics[width=0.4\textwidth]{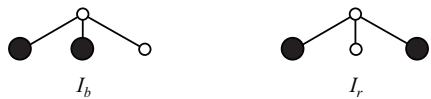}
\end{center}
\vspace{-1em}
\caption{A yes-instance for {\sc ISReconf} under the TJ rule, 
which is a no-instance for the {\sc sliding token} problem.}
\label{fig:rules}
\end{figure}

\subsection{Related and known results}

As the (ordinary) {\sc independent set} problem is a key problem among thousands of NP-complete problems,
{\sc sliding token} plays a very important role since 
several PSPACE-hardness results have been proved using reductions from it.
Indeed, {\sc sliding token} is one of the most well-studied reconfiguration problems.  

	In addition, reconfiguration problems for {\sc independent set} ({\sc ISReconf}, for short) have been studied under different reconfiguration rules, as follows.
\smallskip

	\begin{listing}{aaa}
	\item[$\bullet$]{\em Token Sliding} (TS rule)~\cite{BC09,BKW14,HearnDemaine2005,HearnDemaine2009,KaminskiMedvedevMilanic2012,Wro14}: 
      This rule corresponds to {\sc sliding token}, that is, 
      we can slide a single token only along an edge of a graph.
							\smallskip
	
	\item[$\bullet$] {\em Token Jumping} (TJ rule)~\cite{BKW14,ItoKaminskiOnoSuzukiUeharaYamanaka2014,KaminskiMedvedevMilanic2012,Wro14}: 
      A single token can ``jump'' to any vertex (including non-adjacent one) 
      if it results in an independent set.
							\smallskip
							
	\item[$\bullet$] {\em Token Addition and Removal} (TAR rule)~\cite{BB14,Bon14,IDHPSUU,KaminskiMedvedevMilanic2012,MNRSS13,MNRW14,Wro14}: 
      We can either add or remove a single token at a time 
      if it results in an independent set of cardinality at least a given threshold minus one. 
      Therefore, under the TAR rule, independent sets in the sequence do not have the same cardinality.
	\end{listing}
\smallskip

\noindent
 	We note that the existence of a desired sequence depends deeply on the reconfiguration rules. 
 (See \figurename~\ref{fig:rules} for example.) 
	However, {\sc ISReconf} is PSPACE-complete under any of the three reconfiguration rules for planar graphs~\cite{BC09,HearnDemaine2005,HearnDemaine2009}, for perfect graphs~\cite{KaminskiMedvedevMilanic2012}, and for bounded bandwidth graphs~\cite{Wro14}.
	The PSPACE-hardness implies that, unless ${\rm NP} = {\rm PSPACE}$, there exists an instance of {\sc sliding token} which requires a super-polynomial number of token-slides even in a minimum-length reconfiguration sequence. 
In such a case, tokens should make ``detours'' to avoid violating independence.
(For example, see the token placed on the vertex $w$ in \figurename~\ref{fig:example}(a);
it is moved twice even though $w \in \bfI_b \cap \bfI_r$.)

	We here explain only the results which are strongly related to this paper, that is, {\sc sliding token} on trees;
see the references above for the other results.  
\medskip

\noindent
{\bf Results for TS rule ({\sc sliding token}).}

	Kami\'nski et al.~\cite{KaminskiMedvedevMilanic2012} gave a linear-time algorithm to solve {\sc sliding token} for cographs 
(also known as $P_4$-free graphs). 
	They also showed that, for any yes-instance on cographs, two given independent sets $\bfI_b$ and $\bfI_r$ have a reconfiguration sequence such that no token makes detour. 

	Very recently, Bonsma et al.~\cite{BKW14} proved that {\sc sliding token} can be solved in polynomial time for claw-free graphs. 
	Note that neither cographs nor claw-free graphs contain trees as a (proper) subclass. 
	Thus, the complexity status for trees was open under the TS rule. 
\medskip

\noindent
{\bf Results for trees}.

	In contrast to the TS rule, it is known that {\sc ISReconf} can be solved in linear time under the TJ and TAR rules for even-hole-free graphs~\cite{KaminskiMedvedevMilanic2012}, which include trees. 
	Indeed, the answer is always ``yes'' under the two rules when restricted to even-hole-free graphs.
	Furthermore, tokens never make detours in even-hole-free graphs under the TJ and TAR rules.

	On the other hand, under the TS rule, tokens are required to make detours even in trees.
(See \figurename~\ref{fig:example}.)
	In addition, there are no-instances for trees under TS rule.
(See \figurename~\ref{fig:rules}.)
	These make the problem much more complicated, and we think they are the main reasons why {\sc sliding token} for trees was open, despite the recent intensive algorithmic research on {\sc ISReconf}~\cite{BB14,Bon14,BKW14,ItoKaminskiOnoSuzukiUeharaYamanaka2014,KaminskiMedvedevMilanic2012,MNRW14}. 

\subsection{Our contribution}

	In this paper, we first prove that the {\sc sliding token} problem is solvable in $O(n)$ time for any tree $T$ with $n$ vertices. 
	Therefore, we can conclude that {\sc ISReconf} for trees is in P (indeed, solvable in linear time) under any of the three reconfiguration rules.
	
	It is remarkable that there exists an infinite family of instances on paths for which any reconfiguration sequence requires $\Omega(n^2)$ length, although we can decide it is a yes-instance in $O(n)$ time. 
	As the second result of this paper, we give an $O(n^2)$-time algorithm which finds an actual reconfiguration sequence of length $O(n^2)$ between two given independent sets for a yes-instance.
	
	Since the treewidth of any graph $G$ can be bounded by the bandwidth of $G$, the result of~\cite{Wro14} implies that {\sc sliding token} is PSPACE-complete for bounded treewidth graphs. 
(See~\cite{Bod98} for the definition of treewidth.)
	Thus, there exists an instance on bounded treewidth graphs which requires a super-polynomial number of token-slides even in a minimum-length reconfiguration sequence unless ${\rm NP} = {\rm PSPACE}$.
	Therefore, it is interesting that any yes-instance on a tree, whose treewidth is one, has an $O(n^2)$-length reconfiguration sequence even though trees require to make detours to transform.

\subsection{Technical overview}
\label{subsec:highlight}

	We here explain our main ideas;
formal descriptions will be given later. 

	We say that a token on a vertex $v$ is ``rigid'' under an independent set $\bfI$ of a tree $T$ if it cannot be slid at all, that is, $v \in \bfI^\prime$ holds for {\em any} independent set $\bfI^\prime$ of $T$ which is reconfigurable from $\bfI$.
(For example, four tokens in \figurename~\ref{fig:rules} are rigid.)
	Our algorithm is based on the following two key points. 
	\begin{listing}{aaa}
	\item[(1)] In Lemma~\ref{lem:rigid}, we will give a simple but non-trivial characterization of rigid tokens, based on which we can find all rigid tokens of two given independent sets $\bfI_b$ and $\bfI_r$ in time $O(n)$. 
	Note that, if $\bfI_b$ and $\bfI_r$ have different placements of rigid tokens, then it is a no-instance (Lemma~\ref{lem:step1}).
	\item[(2)] Otherwise, we obtain a forest by deleting the vertices with rigid tokens together with their neighbors (Lemma~\ref{lem:step2}).
	We will prove in Lemma~\ref{lem:secondkey} that the answer is ``yes'' as long as each tree in the forest contains the same number of tokens in $\bfI_b$ and $\bfI_r$. 
	\end{listing}

\section{Preliminaries}

In this section, we introduce some basic terms and notation. 

\subsection{Graph notation}

In the {\sc sliding token} problem, 
we may assume without loss of generality that graphs are simple and connected.
For a graph $G$, we sometimes denote by $V(G)$ and $E(G)$ the vertex set and edge set of $G$, respectively. 

	In a graph $G$, a vertex $w$ is said to be a {\em neighbor} of a vertex $v$ if $\{v,w \} \in E(G)$. 
	For a vertex $v$ in $G$, let $\Nei{G}{v} = \{w \in V(G) \mid \{v, w \}\in E(G) \}$, and let $\Neiclosed{G}{v} = \Nei{G}{v}\cup\{v\}$.
	For a subset $S \subseteq V(G)$, we simply write $\Neiclosed{G}{S} = \bigcup_{v \in S} \Neiclosed{G}{v}$. 
	For a vertex $v$ of $G$, we denote by $\deg_{G}(v)$ the degree of $v$ in $G$, that is, $\deg_{G}(v) = \msize{\Nei{G}{v}}$.
	For a subgraph $G^\prime$ of a graph $G$, we denote by $G \setminus G^\prime$ the subgraph of $G$ induced by the vertices in $V(G) \setminus V(G^\prime)$. 

\begin{figure}[t]
\begin{center}
\includegraphics[width=0.3\textwidth]{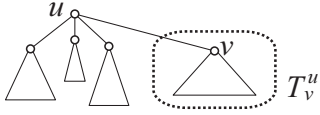}
\end{center}
\vspace{-1em}
\caption{Subtree $\Tsub{v}{u}$ in the whole tree $T$.}
\label{fig:subtree}
\end{figure}

	Let $T$ be a tree. 
	For two vertices $v$ and $w$ in $T$, the unique path between $v$ and $w$ is simply called the {\em $vw$-path} in $T$. 
	We denote by $\dist(v, w)$ the number of edges in the $vw$-path in $T$.
	For two vertices $u$ and $v$ of a tree $T$, let $\Tsub{v}{u}$ be the subtree of $T$ obtained by regarding $u$ as the root of $T$ and then taking the subtree rooted at $v$ which consists of $v$ and all descendants of $v$.
(See \figurename~\ref{fig:subtree}.)
	It should be noted that $u$ is not contained in the subtree $\Tsub{v}{u}$. 

\subsection{Definitions for {\sc sliding token}}

Let $\bfI_i$ and $\bfI_j$ be two independent sets of a graph $G$ such that $\msize{\bfI_i} = \msize{\bfI_j}$. 
If there exists exactly one edge $\{u,v\}$ in $G$ such that $\bfI_{i} \setminus\bfI_{j}=\{u\}$ 
and $\bfI_{j}\setminus\bfI_{i}=\{v\}$, 
then we say that $\bfI_{j}$ can be obtained from $\bfI_{i}$ by {\em sliding} 
the token on $u \in \bfI_{i}$ to its adjacent vertex $v$ along the edge $\{u,v\}$, 
and denote it by $\bfI_{i} \onestep \bfI_{j}$. 
We note that the tokens are unlabeled, while the vertices in a graph are labeled.
	We sometimes omit to say (the label of) the vertex on which a token is placed, and simply say ``a token in an independent set $\bfI$.''

A {\em reconfiguration sequence} between two independent sets $\bfI_1$ and $\bfI_{\ell}$ of $G$ 
is a sequence $\langle \bfI_1, \bfI_2, \ldots, \bfI_{\ell} \rangle$ of 
independent sets of $G$ such that $\bfI_{i-1} \onestep \bfI_i$ for $i=2, 3, \ldots, \ell$.
We sometimes write $\bfI \in \calS$ if an independent set $\bfI$ of $G$ appears in the reconfiguration sequence $\calS$. 
We write $\bfI_{1} \sevstepT{G} \bfI_{\ell}$ if there exists a reconfiguration sequence $\calS$ between $\bfI_1$ and $\bfI_{\ell}$ such that all independent sets $\bfI \in \calS$ satisfy $\bfI \subseteq V(G)$.
%
%
The {\em length} of a reconfiguration sequence $\calS$ is defined as 
the number of independent sets contained in $\calS$.
For example, 
the length of the reconfiguration sequence in \figurename~\ref{fig:example} is $5$. 

Given two independent sets $\bfI_b$ and $\bfI_r$ of a graph $G$, the {\sc sliding token} problem is to determine whether $\bfI_b \sevstepT{G} \bfI_r$ or not.
We may assume without loss of generality that $\msize{\bfI_b} = \msize{\bfI_r}$; 
otherwise the answer is clearly ``no.'' 
Note that {\sc sliding token} is 
a decision problem asking for the existence of a reconfiguration sequence 
between $\bfI_b$ and $\bfI_r$, and hence it does not ask for an actual reconfiguration sequence. 
We always denote by $\bfI_b$ and $\bfI_r$ the {\em initial} and {\em target} independent sets of $G$, 
respectively.

	\section{Algorithm for Trees}
	\label{sec:algorithm}
	
	In this section, we give the main result of this paper. 
	\begin{theorem} \label{the:tree}
	The {\sc sliding token} problem can be solved in linear time for trees. 
	\end{theorem}
	
	As a proof of Theorem~\ref{the:tree}, we give an $O(n)$-time algorithm which solves {\sc sliding token} for a tree with $n$ vertices. 

	\subsection{Rigid tokens}
	
	In this subsection, we formally define the concept of rigid tokens, and give their nice characterization. 
\smallskip
	
	Let $T$ be a tree, and let $\bfI$ be an independent set of $T$. 
	We say that a token on a vertex $v \in \bfI$ is \emph{$(T, \bfI)$-rigid} if $v \in \bfI^\prime$ holds for \emph{any} independent set $\bfI^\prime$ of $T$ such that $\bfI \sevstepT{T} \bfI^\prime$. 
	Conversely, if a token on a vertex $v \in \bfI$ is not $(T, \bfI)$-rigid, then it is \emph{$(T, \bfI)$-movable};
in other words, there exists an independent set $\bfI^\prime$ such that $v \not\in \bfI^\prime$ and $\bfI \sevstepT{T} \bfI^\prime$. 
	For example, in \figurename~\ref{fig:rigidexample}, the tokens $t_1, t_2, t_3, t_4$ are $(T, \bfI)$-rigid, while the tokens $t_5, t_6, t_7$ are $(T, \bfI)$-movable. 
	Note that, even though $t_6$ and $t_7$ cannot be slid to any neighbor in $T$ under $\bfI$, we can slide them after sliding $t_5$ downward. 
	
	We then extend the concept of rigid/movable tokens to subtrees of $T$. 
	For any subtree $T^\prime$ of $T$, we denote simply $\bfI \cap T^\prime =  \bfI \cap V(T^\prime)$.
	Then, a token on a vertex $v \in \bfI \cap T^\prime$ is \emph{$(T^\prime, \bfI \cap T^\prime)$-rigid} if $v \in J$ holds for \emph{any} independent set $J$ of $T^\prime$ such that $\bfI \cap T^\prime \sevstepT{T^\prime} J$.
	For example, in \figurename~\ref{fig:rigidexample}, tokens $t_6$ and $t_7$ are $(T^\prime, \bfI \cap T^\prime)$-rigid even though they are $(T, \bfI)$-movable in the whole tree $T$. 
	Note that, since independent sets are restricted only to the subtree $T^\prime$, we cannot use any vertex (and hence any edge) in $T \setminus T^\prime$ during the reconfiguration. 
	Furthermore, the vertex-subset $J \cup \bigl(\bfI \cap (T \setminus T^\prime) \bigr)$ does not necessarily form an independent set of the whole tree $T$. 
\medskip

\begin{figure}[t]
	\begin{center}
	\includegraphics[width=0.4\textwidth]{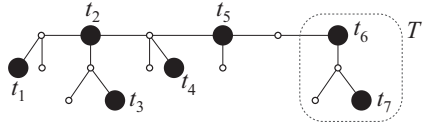}
	\end{center}
	\vspace{-1em}
	\caption{An independent set $\bfI$ of a tree $T$, where $t_1, t_2, t_3, t_4$ are $(T, \bfI)$-rigid tokens and $t_5, t_6, t_7$ are $(T, \bfI)$-movable tokens. For the subtree $T^\prime$, tokens $t_6, t_7$ are $(T^\prime, \bfI \cap T^\prime)$-rigid.}
	\label{fig:rigidexample}
\end{figure}

 \begin{figure}[b]
	\begin{center}
	\includegraphics[width=0.78\textwidth]{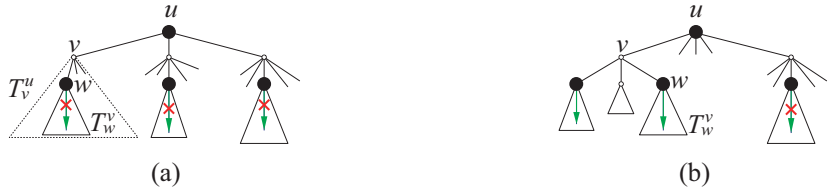}
	\end{center}
	\vspace{-1em}
	\caption{(a) A $(T, \bfI)$-rigid token on $u$, and (b) a $(T, \bfI)$-movable token on $u$.}
	\label{fig:rigidmovable}
\end{figure}

	We now give our first key lemma, which gives a characterization of rigid tokens. 
(See also \figurename~\ref{fig:rigidmovable}(a) for the claim (b) below.)
	\begin{lemma} \label{lem:rigid}
	Let $\bfI$ be an independent set of a tree $T$, and let $u$ be a vertex in $\bfI$.
	\begin{listing}{aaa}
	\item[{\rm (}a{\rm )}] Suppose that $\msize{V(T)} = \msize{\{u\}} = 1$. Then, the token on $u$ is $(T,\bfI)$-rigid.
	\item[{\rm (}b{\rm )}] Suppose that $\msize{V(T)} \ge 2$. 
					Then, a token on $u$ is $(T,\bfI)$-rigid if and only if, 
					for every neighbor $v \in \Nei{T}{u}$, there exists a vertex $w \in \bfI \cap \Nei{\Tsub{v}{u}}{v}$
					such that the token on $w$ is $(\Tsub{w}{v}, \bfI \cap \Tsub{w}{v})$-rigid.
	\end{listing}
	\end{lemma}
	\begin{proof}
	Obviously, the claim (a) holds.
	In the following, we thus assume that $\msize{V(T)} \ge 2$ and prove the claim (b).

	We first show the if-part.
	Suppose that, for every neighbor $v \in \Nei{T}{u}$, there exists a vertex $w \in \bfI \cap \Nei{\Tsub{v}{u}}{v}$ such that the token on $w$ is $(\Tsub{w}{v}, \bfI \cap \Tsub{w}{v})$-rigid.
(See \figurename~\ref{fig:rigidmovable}(a).)
	Then, we will prove that the token $t$ on $u$ is $(T, \bfI)$-rigid. 
	Since we can slide a token only along an edge of $T$, if $t$ is not $(T, \bfI)$-rigid (and hence is $(T, \bfI)$-movable), then it must be slid to some neighbor $v \in \Nei{T}{u}$.
	By the assumption, $v$ is adjacent with another token $t^\prime$ placed on $w \in \bfI \cap \Nei{\Tsub{v}{u}}{v}$, and hence we first have to slide $t^\prime$ to one of its neighbors other than $v$.
	However, this is impossible since the token $t^\prime$ on $w$ is assumed to be $(\Tsub{w}{v}, \bfI \cap \Tsub{w}{v})$-rigid and hence $w \in J$ holds for any independent set $J$ of $\Tsub{w}{v}$ such that $\bfI \cap \Tsub{w}{v} \sevstepT{\Tsub{w}{v}} J$. 
	We can thus conclude that $t$ is $(T, \bfI)$-rigid. 

	We then show the only-if-part by taking a contrapositive. 
	Suppose that $u$ has a neighbor $v \in \Nei{T}{u}$ such that either $\bfI \cap \Nei{\Tsub{v}{u}}{v} = \emptyset$ or all tokens on $w \in \bfI \cap \Nei{\Tsub{v}{u}}{v}$ are $(\Tsub{w}{v}, \bfI \cap \Tsub{w}{v})$-movable.
(See \figurename~\ref{fig:rigidmovable}(b).)
	Then, we will prove that the token $t$ on $u$ is $(T, \bfI)$-movable;
in particular, we can slide $t$ from $u$ to $v$. 
	Since any token $t^\prime$ on a vertex $w \in \bfI \cap \Nei{\Tsub{v}{u}}{v}$ is $(\Tsub{w}{v}, \bfI \cap \Tsub{w}{v})$-movable, we can slide $t^\prime$ to some vertex in $\Tsub{w}{v}$ via a reconfiguration sequence $\calS_{w}$ in $\Tsub{w}{v}$. 
	Recall that only the vertex $v$ is adjacent with a vertex in $\Tsub{w}{v}$ and $v \not\in \bfI$.
	Therefore, $\calS_{w}$ can be naturally extended to a reconfiguration sequence $\calS$ in the whole tree $T$ such that $\bfI^\prime \cap \bigl( T \setminus \Tsub{w}{v} \bigr) = \bfI \cap \bigl( T \setminus \Tsub{w}{v} \bigr)$ holds for any independent set $\bfI^\prime \in \calS$ of $T$. 
	Apply this process to all tokens on vertices in $\bfI \cap \Nei{\Tsub{v}{u}}{v}$, and obtain an independent set $\bfI^{\prime \prime}$ of $T$ such that $\bfI^{\prime \prime} \cap \Nei{\Tsub{v}{u}}{v} = \emptyset$.
	Then, we can slide the token $t$ on $u$ to $v$. 
	Thus, $t$ is $(T, \bfI)$-movable. 
\qed
%
\end{proof}



	The following lemma is useful for proving the correctness of our algorithm in Section~\ref{subsec:correctness}.
	\begin{lemma} \label{lem:atmostone}
	Let $\bfI$ be an independent set of a tree $T$ such that all tokens are $(T, \bfI)$-movable, and let $v$ be a vertex such that $v \not\in \bfI$. 
	Then, there exists at most one neighbor $w \in \bfI \cap \Nei{T}{v}$ such that the token on $w$ is $(\Tsub{w}{v}, \bfI \cap \Tsub{w}{v})$-rigid. 
	\end{lemma}
	\begin{proof}
	Suppose for a contradiction that there exist two neighbors $w$ and $w^\prime$ in $\bfI \cap \Nei{T}{v}$ such that the tokens on $w$ and $w^\prime$ are $(\Tsub{w}{v}, \bfI \cap \Tsub{w}{v})$-rigid and $(\Tsub{w^\prime}{v}, \bfI \cap \Tsub{w^\prime}{v})$-rigid, respectively.
(See \figurename~\ref{fig:atmostone}.)
	Since the token $t$ on $w$ is $(\Tsub{w}{v}, \bfI \cap \Tsub{w}{v})$-rigid but is $(T,\bfI)$-movable, there is a reconfiguration sequence $\calS_t$ starting from $\bfI$ which slides $t$ to $v$. 
	However, before sliding $t$ to $v$, $\calS_t$ must slide the token $t^\prime$ on $w^\prime$ to some vertex in $\Nei{\Tsub{w^\prime}{v}}{w^\prime}$. 
	This contradicts the assumption that $t^\prime$ is $(\Tsub{w^\prime}{v}, \bfI \cap \Tsub{w^\prime}{v})$-rigid.
	\qed
	\end{proof}
\begin{figure}[t]
	\begin{center}
	\includegraphics[width=0.25\textwidth]{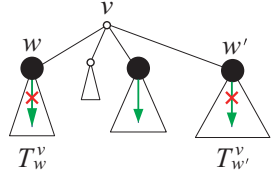}
	\end{center}
	\vspace{-1em}
	\caption{Illustration for Lemma~\ref{lem:atmostone}.}
	\label{fig:atmostone}
\end{figure}

	\subsection{Linear-time algorithm}
	\label{subsec:algorithm}
	
	In this subsection, we describe an algorithm to solve the {\sc sliding token} problem for trees, and estimate its running time;
the correctness of the algorithm will be proved in Section~\ref{subsec:correctness}. 
\smallskip

	Let $T$ be a tree with $n$ vertices, and let $\bfI_b$ and $\bfI_r$ be two given independent sets of $T$. 
	For an independent set $\bfI$ of $T$, we denote by $\ImSet{}$ the set of all vertices in $\bfI$ on which $(T, \bfI)$-rigid tokens are placed. 
	Then, the following algorithm determines whether $\bfI_b \sevstepT{T} \bfI_r$ or not.
\smallskip

	\begin{listing}{{\bf Step~2.}}
	\item[{\bf Step~1.}] Compute $\ImSet{b}$ and $\ImSet{r}$.
								Return ``no'' if $\ImSet{b} \neq \ImSet{r}$; otherwise go to Step 2.
\smallskip

	\item[{\bf Step~2.}] Delete the vertices in $\Neiclosed{T}{\ImSet{b}} = \Neiclosed{T}{\ImSet{r}}$ from $T$, and obtain a forest $\Forest$ consisting of $\numq$ trees $T_1, T_2, \ldots, T_{\numq}$. 
								Return ``yes'' if $\msize{\bfI_b \cap T_j} = \msize{\bfI_r \cap T_j}$ holds for every $j \in \{1, 2, \ldots, \numq \}$; otherwise return ``no.''
	\end{listing}
\smallskip

	We now show that our algorithm above runs in $O(n)$ time. 
	Clearly, Step~2 can be done in $O(n)$ time, and hence we will show that Step~1 can be executed in $O(n)$ time.
\smallskip
	
	We first give the following property of rigid tokens on a tree, which says that deleting movable tokens does not affect the rigidity of the other tokens.
	\begin{lemma} \label{lem:removing-movable-tokens}
	Let $\bfI$ be an independent set of a tree $T$.
	Assume that the token on a vertex $x \in \bfI$ is $(T, \bfI)$-movable.
	Then, for every vertex $u \in I \setminus \{x\}$, the token on $u$ is $(T,\bfI)$-rigid if and only if it is $(T,\bfI \setminus \{x\})$-rigid.
	\end{lemma}
	\begin{proof}
	The if-part is trivially true, because we cannot make a rigid token movable by adding another token.
	We thus show the only-if-part by contradiction.

	Let $\bfI^\prime = \bfI \setminus \{x\}$.
	Suppose that $u \in \bfI$ is the closest vertex to $x$ such that its token is $(T,\bfI)$-rigid but $(T,\bfI^\prime)$-movable.
	We assume that $x$ is contained in a subtree $\Tsub{v}{u}$ for a neighbor $v$ of $u$. 
(See \figurename~\ref{fig:removing_movable_tokens}.)
	Note that $x \neq v$ since $x, u \in \bfI$.
	Since the token $t_u$ on $u$ is $(T, \bfI)$-rigid, by Lemma~\ref{lem:rigid} the vertex $v \in \Nei{T}{u}$ has at least one neighbor $w \in \bfI \cap \Nei{\Tsub{v}{u}}{v}$ such that the token $t_w$ on $w$ is $(\Tsub{w}{v}, \bfI \cap \Tsub{w}{v})$-rigid.
	Indeed, $t_w$ is $(T, \bfI)$-rigid, because $t_u$ is assumed to be $(T, \bfI)$-rigid.
	Thus, we know that $x \neq w$ since the token $t_x$ on $x$ is $(T, \bfI)$-movable.

	First, consider the case where $x$ is contained in a subtree $\Tsub{w^\prime}{v}$ for some neighbor $w^\prime$ of $v$ other than $w$. 
(See \figurename~\ref{fig:removing_movable_tokens}(a).)
	Then, $\bfI^\prime \cap \Tsub{w}{v} = \bfI \cap \Tsub{w}{v}$.
	Since $t_w$ is $(\Tsub{w}{v}, \bfI \cap \Tsub{w}{v})$-rigid, it is also $(\Tsub{w}{v}, \bfI^\prime \cap \Tsub{w}{v})$-rigid.
	Therefore, by Lemma~\ref{lem:rigid} the token $t_u$ is $(T,\bfI^\prime)$-rigid.
	This contradicts the assumption that $t_u$ is $(T,\bfI^\prime)$-movable.


	We thus consider the case where $x \in V(\Tsub{w}{v}) \setminus \{w\}$. 
(See \figurename~\ref{fig:removing_movable_tokens}(b).)
	Recall that $\bfI^\prime$ is obtained by deleting only $x$ from $\bfI$. 
	Then, since $t_u$ is $(T, \bfI)$-rigid but $(T, \bfI^\prime)$-movable, it must be slid from $u$ to $v$. 
	However, before executing this token-slide, we have to slide $t_w$ to some vertex in $\Nei{\Tsub{w}{v}}{w}$.
	Thus, $t_w$ is $(\Tsub{w}{v}, \bfI^\prime \cap \Tsub{w}{v})$-movable, and hence it is also $(T, \bfI^\prime)$-movable.
	Since $t_w$ is $(T, \bfI)$-rigid and $w$ is strictly closer to $x \in V(\Tsub{w}{v})$ than $u$, this contradicts the assumption that $u$ is the closest vertex to $x$ such that its token is $(T,\bfI)$-rigid but $(T,\bfI^\prime)$-movable.
%
%
%
%
%
%
\qed
\end{proof}

\begin{figure}[t]
	\begin{center}
	\includegraphics[width=0.75\textwidth]{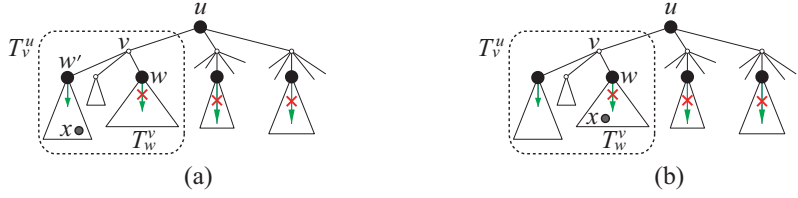}
	\end{center}
	\vspace{-1em}
	\caption{Illustration for Lemma~\ref{lem:removing-movable-tokens}.}
	\label{fig:removing_movable_tokens}
\end{figure}

	Then, the following lemma proves that Step~1 can be executed in $O(n)$ time.  
	\begin{lemma} \label{lem:all-rigid-tokens}
	For an independent set $\bfI$ of a tree $T$ with $n$ vertices, $\ImSet{}$ can be computed in $O(n)$ time.
	\end{lemma}
	\begin{proof}
	Lemma~\ref{lem:removing-movable-tokens} implies that the set $\ImSet{}$ of all $(T, \bfI)$-rigid tokens in $\bfI$ can be found by removing all $(T, \bfI)$-movable tokens in $\bfI$.
	Observe that, if $\bfI$ contains $(T, \bfI)$-movable tokens, then at least one of them can be immediately slid to one of its neighbors.
	That is, there is a token on $u \in \bfI$ which has a neighbor $w \in \Nei{T}{u}$ such that $\Nei{T}{w} \cap \bfI = \{u\}$.
	Then, the following algorithm efficiently finds and removes such tokens iteratively.
	\smallskip

	\begin{listing}{{\bf Step~B.}}
	\item[{\bf Step~A.}] Define and compute $\deg_{\bfI}(w) = |\Nei{T}{w} \cap \bfI|$ for all vertices $w \in V(T)$.
	\smallskip

	\item[{\bf Step~B.}] Define and compute $M = \{u \in \bfI \mid \exists w \in \Nei{T}{u} \mbox{ such that }\deg_{\bfI}(w) = 1\}$.
	\smallskip
	       
	\item[{\bf Step~C.}] Repeat the following steps (\one)--(\three) until $M = \emptyset$.
		\begin{listing}{{\bf (\three)}}
		\item[{\bf (\one)}] Select an arbitrary vertex $u \in M$, and remove it from $M$ and $\bfI$.
		\item[{\bf (\two)}] Update $\deg_{\bfI}(w) := \deg_{\bfI}(w) - 1$ for each neighbor $w \in \Nei{T}{u}$.
		\item[{\bf (\three)}]  If $\deg_{\bfI}(w)$ becomes one by the update (\two) above, then add the vertex $u^\prime \in \Nei{T}{w} \cap \bfI$ into $M$.
		\end{listing}
  \smallskip

  \item[{\bf Step~D.}] Output $\bfI$ as the set $\ImSet{}$.
\end{listing}
\smallskip

	Clearly, Steps A, B and D can be done in $O(n)$ time.
	We now show that Step~C takes only $O(n)$ time.
	Each vertex in $\bfI$ can be selected at most once as $u$ at Step~C-(\one).
	For the selected vertex $u$, Step~C-(\two) takes $O(\deg_{T}(u))$ time for updating $\deg_{I}(w)$ of its neighbors $w \in \Nei{T}{u}$.
	Each vertex in $V(T) \setminus \bfI$ can be selected at most once as $w$ at Step~C-(\three).
	For the selected vertex $w$, Step~C-(\three) takes $O(\deg_{T}(w))$ time for finding $u^\prime \in \Nei{T}{w} \cap \bfI$.
	Therefore, Step~C takes $O \Bigl(\sum_{v \in V(T)} \deg_{T}(v) \Bigr) = O(n)$ time in total.
\qed
\end{proof}

	Therefore, Step~1 of our algorithm can be done in $O(n)$ time, and hence the algorithm runs in linear time in total.

	\subsection{Correctness of the algorithm}
	\label{subsec:correctness}

	In this subsection, we prove that the $O(n)$-time algorithm in Section~\ref{subsec:algorithm} correctly determines whether $\bfI_b \sevstepT{T} \bfI_r$ or not, for two given independent sets $\bfI_b$ and $\bfI_r$ of a tree $T$.
\smallskip
	
	We first show the correctness of Step~1.
	\begin{lemma} \label{lem:step1}
	Suppose that $\ImSet{b} \neq \ImSet{r}$ for two given independent sets $\bfI_b$ and $\bfI_r$ of a tree $T$. 
	Then, it is a no-instance.
	\end{lemma}
	\begin{proof}
	By the definition of rigid tokens, $\ImSet{b} = \ImSetp$ holds for \emph{any} independent set $\bfI^\prime$ of $T$ such that $\bfI_b \sevstepT{T} \bfI^\prime$. 
	Therefore, there is no reconfiguration sequence between $\bfI_b$ and $\bfI_r$ if $\ImSet{r} \neq \ImSet{b}$.
	\qed
	\end{proof}

	We then show the correctness of Step~2. 
	We first claim that deleting the vertices with rigid tokens together with their neighbors does not affect the reconfigurability. 
	\begin{lemma} \label{lem:step2}
	Suppose that $\ImSet{b} = \ImSet{r}$ for two given independent sets $\bfI_b$ and $\bfI_r$ of a tree $T$, and let $\Forest$ be the forest obtained by deleting the vertices in $\Neiclosed{T}{\ImSet{b}} = \Neiclosed{T}{\ImSet{r}}$ from $T$.
	Then, $\bfI_b \sevstepT{T} \bfI_r$ if and only if $\bfI_b \cap \Forest \sevstepT{\Forest} \bfI_r \cap \Forest$.
	Furthermore, all tokens in $\bfI_b \cap \Forest$ are $(\Forest, \bfI_b \cap \Forest)$-movable, and all tokens in $\bfI_r \cap \Forest$ are $(\Forest, \bfI_r \cap \Forest)$-movable.
	\end{lemma}
	\begin{proof}
	We first prove the if-part. 
	Suppose that $\bfI_b \cap \Forest \sevstepT{\Forest} \bfI_r \cap \Forest$, and hence there exists a reconfiguration sequence $\calS_{\Forest}$ between $\bfI_b \cap \Forest$ and $\bfI_r \cap \Forest$. 
	Then, for each independent set $\bfI \in \calS_{\Forest}$ of $\Forest$, the vertex-subset $\ImSet{b} \cup \bfI = \ImSet{r} \cup \bfI$ forms an independent set of $T$ since $\Forest$ is obtained by deleting all vertices in $\Neiclosed{T}{\ImSet{b}} = \Neiclosed{T}{\ImSet{r}}$. 
	Therefore, $\calS_{\Forest}$ can be extended to a reconfiguration sequence between $\bfI_b$ and $\bfI_r$ of $T$. 
	We thus have $\bfI_b \sevstepT{T} \bfI_r$. 
	
	We then prove the only-if-part. 
	Suppose that $\bfI_b \sevstepT{T} \bfI_r$, and hence there exists a reconfiguration sequence $\calS_T$ between $\bfI_b$ and $\bfI_r$. 
	Then, for any independent set $\bfI \in \calS_T$, we have $\bfI_b \sevstepT{T} \bfI$ and $\bfI \sevstepT{T} \bfI_r$, and hence by the definition of rigid tokens $\ImSet{b} = \ImSet{r} \subseteq \bfI$ holds. 
	Furthermore, $\bfI \setminus \ImSet{b} = \bfI \setminus \ImSet{r}$ is a vertex-subset of $V(\Forest)$ since no token can be placed on any neighbor of $\ImSet{b} = \ImSet{r}$. 
	Therefore, $\bfI \setminus \ImSet{b} = \bfI \setminus \ImSet{r}$ forms an independent set of $\Forest$. 
	For two consecutive independent sets $\bfI_{i-1}$ and $\bfI_i$ in $\calS_T$, let $\bfI_{i-1} \setminus \bfI_{i} = \{u\}$ and $\bfI_{i} \setminus \bfI_{i-1} = \{v\}$. 
	Since $u \notin \bfI_{i}$ and $v \notin \bfI_{i-1}$, neither $u$ nor $v$ are in $\ImSet{b} = \ImSet{r}$.
	Therefore, we have $u, v \in V(\Forest)$, and hence the edge $\{u, v\}$ is in $E(\Forest)$. 
	Then, we can obtain a reconfiguration sequence between $\bfI_b \cap \Forest$ and $\bfI_r \cap \Forest$ by replacing all independent sets $\bfI \in \calS_T$ with $\bfI \cap \Forest$. 
	We thus have $\bfI_b \cap \Forest \sevstepT{\Forest} \bfI_r \cap \Forest$. 
\smallskip
	
	We finally prove that all tokens in $\bfI_b \cap \Forest$ are $(\Forest, \bfI_b \cap \Forest)$-movable.
(The proof for the tokens in $\bfI_r \cap \Forest$ is the same.)
	Notice that each token $t$ on a vertex $v$ in $\bfI_b \cap \Forest$ is $(T, \bfI_b)$-movable; otherwise $t \in \ImSet{b}$. 
	Therefore, there exists an independent set $\bfI^\prime$ of $T$ such that $v \not\in \bfI^\prime$ and $\bfI_b \sevstepT{T} \bfI^\prime$. 
	Then, $\bfI_b \cap \Forest \sevstepT{\Forest} \bfI^\prime \cap \Forest$ as we have proved above, and hence $t$ is $(\Forest, \bfI_b \cap \Forest)$-movable.
	\qed
	\end{proof}

	Suppose that $\ImSet{b} = \ImSet{r}$ for two given independent sets $\bfI_b$ and $\bfI_r$ of a tree $T$.
	Let $\Forest$ be the forest consisting of $\numq$ trees $T_1, T_2, \ldots, T_{\numq}$, which is obtained from $T$ by deleting the vertices in $\Neiclosed{T}{\ImSet{b}} = \Neiclosed{T}{\ImSet{r}}$. 
	Since we can slide a token only along an edge of $\Forest$, we clearly have $\bfI_b \cap \Forest \sevstepT{\Forest} \bfI_r \cap \Forest$ if and only if $\bfI_b \cap T_j \sevstepT{T_j} \bfI_r \cap T_j$ for all $j \in \{1, 2, \ldots, \numq \}$.
	Furthermore, Lemma~\ref{lem:step2} implies that, for each $j \in \{1,2, \ldots, \numq\}$, all tokens in $\bfI_b \cap T_j$ are $(T_j, \bfI_b \cap T_j)$-movable;
similarly, all tokens in $\bfI_r \cap T_j$ are $(T_j, \bfI_r \cap T_j)$-movable. 
\medskip

	We now give our second key lemma, which completes the correctness proof of our algorithm. 
	\begin{lemma} \label{lem:secondkey}
	Let $\bfI_b$ and $\bfI_r$ be two independent sets of a tree $T$ such that all tokens in $\bfI_b$ and $\bfI_r$ are $(T, \bfI_b)$-movable and $(T, \bfI_r)$-movable, respectively.
	Then, $\bfI_b \sevstepT{T} \bfI_r$ if and only if $\msize{\bfI_b} = \msize{\bfI_r}$. 
	\end{lemma}
	
	The only-if-part of Lemma~\ref{lem:secondkey} is trivial, and hence we prove the if-part.
	In our proof, we do \emph{not} reconfigure $\bfI_b$ into $\bfI_r$ directly, but reconfigure both $\bfI_b$ and $\bfI_r$ into some independent set $\bfIstar$ of $T$. 
	Note that, since any reconfiguration sequence is reversible, $\bfI_b \sevstepT{T} \bfIstar$ and $\bfI_r \sevstepT{T} \bfIstar$ imply that $\bfI_b \sevstepT{T} \bfI_r$.

	We say that a degree-$1$ vertex $v$ of $T$ is \emph{safe} if its unique neighbor $u$ has at most one neighbor $w$ of degree more than one.
(See \figurename~\ref{fig:safe}.) 
	Note that any tree has at least one safe degree-$1$ vertex. 
\begin{figure}[t]
	\begin{center}
	\includegraphics[width=0.23\textwidth]{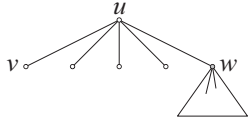}
	\end{center}
	\vspace{-1em}
	\caption{A degree-$1$ vertex $v$ of a tree $T$ which is safe.}
	\label{fig:safe}
\end{figure}

	As the first step of the if-part proof, we give the following lemma. 
	\begin{lemma} \label{lem:degree1}
	Let $\bfI$ be an independent set of a tree $T$ such that all tokens in $\bfI$ are $(T, \bfI)$-movable, and let $v$ be a safe degree-$1$ vertex of $T$. 
	Then, there exists an independent set $\bfIp$ such that $v \in \bfIp$ and $\bfI \sevstepT{T} \bfIp$.
	\end{lemma}
	\begin{proof}
	Suppose that $v \not\in \bfI$; otherwise the lemma clearly holds. 
	We will show that one of the closest tokens from $v$ can be slid to $v$.
	Let $M = \{w \in \bfI \mid \dist(v,w) = \min_{x \in \bfI} \dist(v,x)\}$.
	Let $w$ be an arbitrary vertex in $M$, and let $P = (p_{0} = v, p_{1}, \dots, p_{\ell} = w)$ be the $vw$-path in $T$.
(See \figurename~\ref{fig:nearest}.)
	If $\ell = 1$ and hence $p_1 \in \bfI$, then we can simply slide the token on $p_1$ to $v$. 
	Thus, we may assume that $\ell \ge 2$. 

	We note that no token is placed on the vertices $p_{0}, \dots, p_{\ell-1}$ and the neighbors of $p_{0}, \dots, p_{\ell-2}$, 
because otherwise the token on $w$ is not closest to $v$. 
	Let $M^\prime = M \cap \Nei{T}{p_{\ell-1}}$.
	Since $p_{\ell -1} \not\in \bfI$, by Lemma~\ref{lem:atmostone} there exists at most one vertex $w^\prime \in M^\prime$ such that the token on $w^\prime$ is $(\Tsub{w^\prime}{p_{\ell-1}}, \bfI \cap \Tsub{w^\prime}{p_{\ell-1}})$-rigid. 
	We choose such a vertex $w^\prime$ if it exists, otherwise choose an arbitrary vertex in $M^\prime$ and regard it as $w^\prime$. 
	
	Since all tokens on the vertices $w^{\prime \prime}$ in $M^\prime \setminus \{w^\prime\}$ are $(\Tsub{w^{\prime \prime}}{p_{\ell-1}}, \bfI \cap \Tsub{w^{\prime \prime}}{p_{\ell-1}})$-movable, we first slide the tokens on $w^{\prime \prime}$ to some vertices in $\Tsub{w^{\prime \prime}}{p_{\ell-1}}$.
	Then, we can slide the token on $w^\prime$ to $v$ along the path $P$. 
	In this way, we can obtain an independent set $\bfI^\prime$ such that $v \in \bfIp$ and $\bfI \sevstepT{T} \bfIp$.
	\qed
	\end{proof}
\begin{figure}[t]
	\begin{center}
	\includegraphics[width=0.35\textwidth]{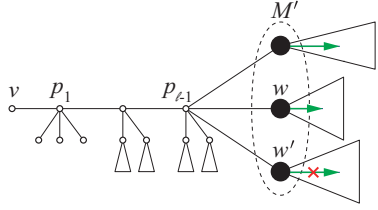}
	\end{center}
	\vspace{-1em}
	\caption{Illustration for Lemma~\ref{lem:degree1}.}
	\label{fig:nearest}
\end{figure}

	We then prove that deleting a safe degree-$1$ vertex with a token does not affect the movability of the other tokens. 
(See also \figurename~\ref{fig:delete1}.) 
	\begin{lemma} \label{lem:delete1}
	Let $v$ be a safe degree-$1$ vertex of a tree $T$, and let $\Tminus$ be the subtree of $T$ obtained by deleting $v$, its unique neighbor $u$, and the resulting isolated vertices.
	Let $\bfI$ be an independent set of $T$ such that $v \in \bfI$ and all tokens are $(T, \bfI)$-movable.
	Then, all tokens in $\bfI \setminus \{v\}$ are $(\Tminus, \bfI \setminus \{v\})$-movable.
	\end{lemma}
	\begin{proof}
	Since $\Tsub{v}{u}$ consists of a single vertex $v$, the token on $v$ is $(\Tsub{v}{u}, \bfI \cap \Tsub{v}{u})$-rigid. 
	Therefore, no token is placed on degree-$1$ neighbors of $u$ other than $v$ (see \figurename~\ref{fig:delete1}), because otherwise it contradicts to Lemma~\ref{lem:atmostone};
recall that all tokens in $\bfI$ are assumed to be $(T, \bfI)$-movable.

%

	Let $\Iminus = \bfI \setminus \{ v \}$. 
	Suppose for a contradiction that there exists a token in $\Iminus$ which is $(\Tminus, \Iminus)$-rigid.
	Let $w_p \in \Iminus$ be such a vertex closest to $v$, and let $z$ be the vertex on the $vw_p$-path right before $w_p$.
\smallskip

\begin{figure}[b]
	\begin{center}
	\includegraphics[width=0.85\textwidth]{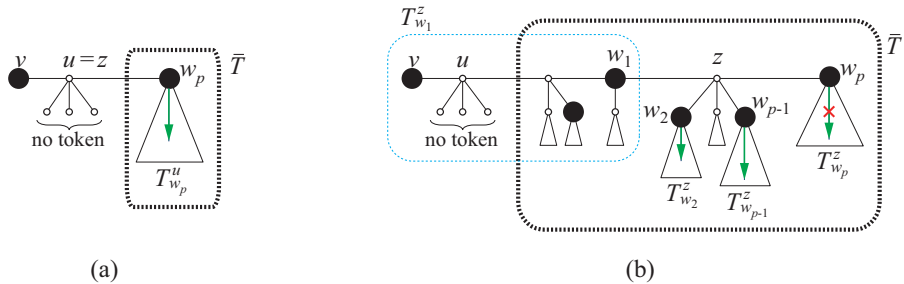}
	\end{center}
	\vspace{-1em}
	\caption{Illustration for Lemma~\ref{lem:delete1}.}
	\label{fig:delete1}
\end{figure}

\noindent
	{\bf Case (1):} $z = u$. (See \figurename~\ref{fig:delete1}(a).)

	Recall that the token on $v$ is $(T, \bfI)$-movable, but is $(\Tsub{v}{u}, \bfI \cap \Tsub{v}{u})$-rigid. 
	Therefore, by Lemma~\ref{lem:atmostone} the token on $w_p$ must be $(\Tsub{w_p}{u}, \bfI \cap \Tsub{w_p}{u})$-movable.
	However, this contradicts the assumption that $w_p$ is $(\Tminus, \Iminus)$-rigid, because $\Tminus = \Tsub{w_p}{u}$ and $\Iminus = \bfI \cap \Tsub{w_p}{u}$ in this case. 
\medskip

\noindent
	{\bf Case (2):} $z \neq u$. (See \figurename~\ref{fig:delete1}(b).)

	Let $w_1$ be the neighbor of $z$ on the $vw_p$-path other than $w_p$; 
let $\Nei{T}{z} = \{ w_1, w_2, \ldots, w_p\}$.
	We note that the subtree $\Tsub{w_1}{z}$ contains the deleted star $T \setminus \Tminus$ centered at $u$.
	
	We first note that the token $t_p$ on $w_p$ is $(\Tmsub{w_p}{z}, \Iminus \cap \Tmsub{w_p}{z})$-rigid, because otherwise $t_p$ can be slid to some vertex in $\Tmsub{w_p}{z}$ and hence it is $(\Tminus, \Iminus)$-movable. 
	Since $\Tmsub{w_p}{z} = \Tsub{w_p}{z}$ and $\Iminus \cap \Tmsub{w_p}{z} = \bfI \cap \Tsub{w_p}{z}$, the token $t_p$ is also $(\Tsub{w_p}{z}, \bfI \cap \Tsub{w_p}{z})$-rigid. 
	
	For each $j \in \{2,3,\ldots, p-1\}$ with $w_j \in \bfI$, since $t_p$ is $(\Tsub{w_p}{z}, \bfI \cap \Tsub{w_p}{z})$-rigid and all tokens in $\bfI$ are $(T, \bfI)$-movable, by Lemma~\ref{lem:atmostone} each token $t_j$ on $w_j$ is $(\Tsub{w_j}{z}, \bfI \cap \Tsub{w_j}{z})$-movable. 
	Then, since $\Tsub{w_j}{z} = \Tmsub{w_j}{z}$ and $\bfI \cap \Tsub{w_j}{z} = \Iminus \cap \Tmsub{w_j}{z}$, the token $t_j$ is $(\Tmsub{w_j}{z}, \Iminus \cap \Tmsub{w_j}{z})$-movable. 
	Therefore, if $w_1 \not\in \Iminus$ or the token $t_1$ on $w_1$ is $(\Tmsub{w_1}{z}, \Iminus \cap \Tmsub{w_1}{z})$-movable, then we can slide $t_p$ from $w_p$ to $z$ after sliding each token $t_j$ in $\Iminus \cap \{ w_1, w_2, \ldots, w_{p-1} \}$ to some vertex of the subtree $\Tmsub{w_j}{z}$.
	This contradicts the assumption that $t_p$ is $(\Tminus, \Iminus)$-rigid. 
	
	Therefore, we have $w_1 \in \Iminus$ and a token $t_1$ on $w_1$ is $(\Tmsub{w_1}{z}, \Iminus \cap \Tmsub{w_1}{z})$-rigid. 
	However, since $t_p$ is $(\Tmsub{w_p}{z}, \Iminus \cap \Tmsub{w_p}{z})$-rigid, this implies that $t_1$ is $(\Tminus, \Iminus)$-rigid. 
	Since $w_1$ is on the $vw_p$-path in $T$, this contradicts the assumption that $t_p$ is the $(\Tminus, \Iminus)$-rigid token closest to $v$.  
\qed
\end{proof}

\smallskip

\noindent
	{\bf Proof of the if-part of Lemma~\ref{lem:secondkey}.}
	
	We now prove the if-part of the lemma by the induction on the number of tokens $\msize{\bfI_b} = \msize{\bfI_r}$. 
	The lemma clearly holds for any tree $T$ if $\msize{\bfI_b} = \msize{\bfI_r} = 1$, because $T$ has only one token and hence we can slide it along the unique path in $T$. 
	
	We choose an arbitrary safe degree-$1$ vertex $v$ of a tree $T$, whose unique neighbor is $u$. 
	Since all tokens in $\bfI_b$ are $(T, \bfI_b)$-movable, by Lemma~\ref{lem:degree1} we can obtain an independent set $\bfI_b^\prime$ of $T$ such that $v \in \bfI_b^\prime$ and $\bfI_b \sevstepT{T} \bfI_b^\prime$. 
	By Lemma~\ref{lem:delete1} all tokens in $\bfI_b^\prime \setminus \{ v \}$ are $(\Tminus, \bfI_b^\prime \setminus \{v\})$-movable, where $\Tminus$ is the subtree defined in Lemma~\ref{lem:delete1}. 
	Similarly, we can obtain an independent set $\bfI_r^\prime$ of $T$ such that $v \in \bfI_r^\prime$, $\bfI_r \sevstepT{T} \bfI_r^\prime$ and all tokens in $\bfI_r^\prime \setminus \{ v \}$ are $(\Tminus, \bfI_r^\prime \setminus \{v\})$-movable. 
	Apply the induction hypothesis to the pair of independent sets $\bfI_b^\prime \setminus \{ v\}$ and $\bfI_r^\prime \setminus \{v\}$ of $\Tminus$.
	Then, we have $\bfI_b^\prime \setminus \{ v\} \sevstepT{\Tminus} \bfI_r^\prime \setminus \{v\}$.
	Recall that both $u \not\in \bfI_b^\prime$ and $u \not\in \bfI_r^\prime$ hold, and $u$ is the unique neighbor of $v$ in $T$.
	Furthermore, $u \not\in V(\Tminus)$. 
	Therefore, we can extend the reconfiguration sequence in $\Tminus$ between $\bfI_b^\prime \setminus \{v\}$ and $\bfI_r^\prime \setminus \{v\}$ to a reconfiguration sequence in $T$ between $\bfI_b^\prime$ and $\bfI_r^\prime$. 
	We thus have $\bfI_b \sevstepT{T} \bfI_r$. 
	
	This completes the proof of Lemma~\ref{lem:secondkey}, and hence completes the proof of Theorem~\ref{the:tree}.
\qed

	\subsection{Length of reconfiguration sequence}
	\label{subsec:length}
	
	In this subsection, we show that an actual reconfiguration sequence can be found for a yes-instance on trees, by implementing our proofs in Section~\ref{subsec:algorithm}.
	Furthermore, the length of the obtained reconfiguration sequence is at most quadratic. 
	\begin{theorem} \label{the:length}
	Let $\bfI_b$ and $\bfI_r$ be two independent sets of a tree $T$ with $n$ vertices. 
	If $\bfI_b \sevstepT{T} \bfI_r$, then there exists a reconfiguration sequence of length $O(n^2)$ between $\bfI_b$ and $\bfI_r$, and it can be output in $O(n^2)$ time. 
	\end{theorem}
	
	We note that a reconfiguration sequence $\calS$ can be represented by a sequence of edges on which tokens are slid. 
	Therefore, the space for representing $\calS$ can be bounded by a linear in the length of $\calS$. 
	
	By Theorem~\ref{the:tree} we can determine whether $\bfI_b \sevstepT{T} \bfI_r$ or not in $O(n)$ time. 
	In the following, we thus assume that $\bfI_b \sevstepT{T} \bfI_r$. 
	Furthermore, suppose that all tokens in $\bfI_b$ are $(T, \bfI_b)$-movable, and that all tokens in $\bfI_r$ are $(T, \bfI_r)$-movable;
otherwise we obtain the forest by deleting the vertices in $\Neiclosed{T}{\ImSet{b}} = \Neiclosed{T}{\ImSet{r}}$ from $T$, and find a reconfiguration sequence for each tree in the forest, according to Lemma~\ref{lem:step2}. 

	As in the if-part proof of Lemma~\ref{lem:secondkey}, we choose an arbitrary safe degree-$1$ vertex $v$ of $T$, and obtain an independent set $\bfIp_b$ of $T$ such that $v \in \bfIp_b$ and $\bfI_b \sevstepT{T} \bfIp_b$, as follows.
	\begin{listing}{aaa}
	\item[(a)] Find a vertex $w \in \bfI_b$ which is closest to $v$, and let $P = (v, p_{1}, p_{2}, \ldots, p_{\ell-1}, w)$ be the $vw$-path in $T$.
					Let $M^\prime = \bfI_b \cap \Nei{T}{p_{\ell-1}}$.
					(See also \figurename~\ref{fig:nearest}.)
	\item[(b)] Choose a vertex $w^\prime$ such that the token on $w^\prime$ is $(\Tsub{w^\prime}{p_{\ell-1}}, \bfI \cap \Tsub{w^\prime}{p_{\ell-1}})$-rigid if it exists, otherwise choose an arbitrary vertex in $M^\prime$ and regard it as $w^\prime$. 
	\item[(c)] Slide each token on $w^{\prime \prime} \in M^\prime \setminus \{w^\prime\}$ to some vertex in $\Tsub{w^{\prime \prime}}{p_{\ell-1}}$, and then slide the token on $w^\prime$ to $v$.
	\end{listing}
	In Lemma~\ref{lem:degree1} we have proved that such a reconfiguration sequence from $\bfI_b$ to $\bfIp_b$ always exists. 
	We apply the same process to $\bfI_r$, and repeat until we obtain the same independent set $\bfIstar$ of $T$ such that $\bfI_b \sevstepT{T} \bfIstar$ and $\bfI_r \sevstepT{T} \bfIstar$. 
	Note that, since any reconfiguration sequence is reversible, this means that we obtained a reconfiguration sequence between $\bfI_b$ and $\bfI_r$. 
	
	Therefore, to prove Theorem~\ref{the:length}, it suffices to show that the algorithm above runs in $O(n)$ time for one safe degree-$1$ vertex $v$ and the reconfiguration sequence for sliding one token to $v$ is of length $O(n)$. 
	In particular, the following lemma completes the proof of Theorem~\ref{the:length}. 
	\begin{lemma} \label{lem:moving_root}
	Let $\bfI$ be an independent set of a tree $T$, and let $w \in \bfI$. 
	For a neighbor $z \in \Nei{T}{w}$, suppose that the token on $w$ is $(\Tsub{w}{z}, \bfI \cap \Tsub{w}{z})$-movable. 
	Then, there exists a reconfiguration sequence $\calS_w$ of length $O(\msize{V(\Tsub{w}{z})})$ from $\bfI$ to an independent set $\bfIp$ of $T$ such that $w \not\in \bfIp$ and $J \cap (T \setminus \Tsub{w}{z}) = \bfI \cap (T \setminus \Tsub{w}{z})$ for all $J \in \calS_w$. 
	Furthermore, $\calS_w$ can be output in $O(\msize{V(\Tsub{w}{z})})$ time.
	\end{lemma}
	\begin{proof}
	We prove the lemma by the induction on the depth of $\Tsub{w}{z}$, where the depth of a tree is the longest distance from its root to a leaf.
	If the depth of $\Tsub{w}{z}$ is zero (and hence $\Tsub{w}{z}$ consists of a single vertex $w$), then the token on $w$ is $(\Tsub{w}{z}, \bfI \cap \Tsub{w}{z})$-rigid; 
this contradicts the assumption. 
	Therefore, we may assume that the depth is at least one.
	If the depth of $\Tsub{w}{z}$ is exactly one, then $\Tsub{w}{z}$ is a star centered at $w$, and no token is placed on any neighbor of $w$. 
	Thus, we can slide the token on $w$ by $1$ $(< \msize{V(\Tsub{w}{z})})$ token-slides.
	Then, the lemma holds for trees with depth one.

	Assume that the depth of $\Tsub{w}{z}$ is $k \ge 2$, and that the lemma holds for trees with depth at most $k-1$.
	Since $w$ is $(\Tsub{w}{z}, \bfI \cap \Tsub{w}{z})$-movable, by Lemma~\ref{lem:rigid} there is a vertex $y \in \Nei{\Tsub{w}{z}}{w}$ such that all tokens on the vertices $x$ in $\bfI \cap \Nei{\Tsub{y}{w}}{y}$ are $(\Tsub{x}{y}, \bfI \cap \Tsub{x}{y})$-movable. 
(See \figurename~\ref{fig:length}.)
	Then, we can obtain a reconfiguration sequence which (1) first slides all tokens on the vertices $x$ in $\bfI \cap \Nei{\Tsub{y}{w}}{y}$ to some vertices in $\Tsub{x}{y}$, and (2) then slide the token on $w$ to the vertex $y$.
	By applying the induction hypothesis to each subtree $\Tsub{x}{y}$, this reconfiguration sequence is of length 
	\[
		1 + \sum_{x \in \bfI \cap \Nei{\Tsub{y}{w}}{y}} O \left(\msize{V(\Tsub{x}{y})} \right) = O(\msize{V(\Tsub{y}{w})}),
	\]
and can be output in time $O(\msize{V(\Tsub{y}{w})})$. 
	Note that $w \not\in \bfIp$ holds for the obtained independent set $\bfIp$ of $T$.
	Thus, the lemma holds for trees with depth $k$. 
	\qed
	\end{proof}

\begin{figure}[t]
	\begin{center}
	\includegraphics[width=0.38\textwidth]{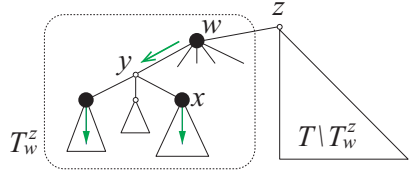}
	\end{center}
	\vspace{-1em}
	\caption{Illustration for Lemma~\ref{lem:moving_root}.}
	\label{fig:length}
\end{figure}

It is interesting that there exists an infinite family of instances on paths
for which any reconfiguration sequence requires $\Omega(n^2)$ length,
where $n$ is the number of vertices. 
For example, consider a path $(v_1,v_2,\ldots,v_{8k})$ with $n=8k$ vertices for any positive integer $k$, and
let $\bfI_b=\{v_1,v_3,v_5,\ldots,v_{2k-1}\}$ and $\bfI_r=\{v_{6k+2},v_{6k+4},\ldots,v_{8k}\}$.
In this yes-instance, any token must be slid $\Theta(n)$ times, 
and hence any reconfiguration sequence requires $\Theta(n^2)$ length to slide them all.

\section{Concluding Remarks}

	In this paper, we have developed an $O(n)$-time algorithm to solve the {\sc sliding token} problem for trees with $n$ vertices, based on a simple but non-trivial characterization of rigid tokens. 
	We have shown that there exists a reconfiguration sequence of length $O(n^2)$ for any yes-instance on trees, and it can be output in $O(n^2)$ time. 
	Furthermore, there exists an infinite family of instances on paths for which any reconfiguration sequence requires $\Omega(n^2)$ length.

	The complexity status of {\sc sliding token} remains open for chordal graphs and interval graphs.
	Interestingly, these graphs have no-instances such that all tokens are movable. 
(See \figurename~\ref{fig:interval} for example.) 

\begin{figure}[t]
	\begin{center}
	\includegraphics[width=0.37\textwidth]{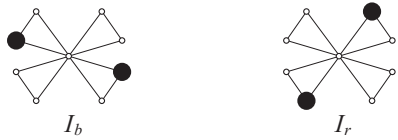}
	\end{center}
	\vspace{-1em}
	\caption{No-instance for an interval graph such that all tokens are movable.}
	\label{fig:interval}
\end{figure}

\bibliographystyle{abbrv}

\begin{thebibliography}{10}

\bibitem{Bod98}
Bodlaender, H.L.:
A partial $k$-arboretum of graphs with bounded treewidth.
Theoretical Computer Science 209, pp.~1--45 (1998)

\bibitem{BB14}
Bonamy, M., Bousquet, N.:
Reconfiguring independent sets in cographs.
{\tt arXiv:1406.1433} (2014)
   
\bibitem{BJLPP14}
Bonamy, M., Johnson, M., Lignos, I., Patel, V., Paulusma, D.: 
Reconfiguration graphs for vertex colourings of chordal and chordal bipartite graphs. 
J.~Combinatorial Optimization 27, pp.~132--143 (2014)

\bibitem{Bon13}
Bonsma, P.: 
The complexity of rerouting shortest paths. 
Theoretical Computer Science 
510, pp.~1--12 (2013)

\bibitem{Bon14}
Bonsma, P.: 
Independent set reconfiguration in cographs. 
To appear in WG 2014, also available at {\tt arXiv:1402.1587} (2014)

\bibitem{BC09}
Bonsma, P., Cereceda, L.: 
Finding paths between graph colourings: PSPACE-completeness and superpolynomial distances. 
Theoretical Computer Science 
410, pp.~5215--5226 (2009)  

\bibitem{BKW14}
Bonsma, P., Kami\'nski, M., Wrochna, M.: 
Reconfiguring independent sets in claw-free graphs.
Proc.~of SWAT 2014, LNCS 8503, pp.~86--97 (2014)


\bibitem{CHJ11}
Cereceda, L., van den Heuvel, J., Johnson, M.:  
Finding paths between 3-colourings. 
J.~Graph Theory 67, pp.~69--82 (2011)

\bibitem{Kolaitis}
Gopalan, P., Kolaitis, P.G., Maneva, E.N., Papadimitriou, C.H.:
The connectivity of Boolean satisfiability: computational and structural dichotomies.
SIAM J.~Computing 38, pp.~2330--2355 (2009)  

\bibitem{HearnDemaine2005}
Hearn, R.A., Demaine, E.D.: 
PSPACE-completeness of sliding-block puzzles and other problems through the nondeterministic constraint logic model of computation. 
Theoretical Computer Science 
343, pp.~72--96 (2005) 

\bibitem{HearnDemaine2009}
Hearn, R.A., Demaine, E.D.: 
Games, Puzzles, and Computation.
A K Peters (2009)

\bibitem{ID11}
Ito, T., Demaine, E.D.:
Approximability of the subset sum reconfiguration problem.
To appear in J.~Combinatorial Optimization, DOI \url{10.1007/s10878-012-9562-z}

\bibitem{IDHPSUU}
Ito, T., Demaine, E.D., Harvey, N.J.A., Papadimitriou, C.H., Sideri, M., Uehara, R., Uno, Y.: 
On the complexity of reconfiguration problems.
Theoretical Computer Science 
412, pp.~1054--1065 (2011)

\bibitem{IKD09}
Ito, T., Kami\'nski, M., Demaine, E.D.: 
Reconfiguration of list edge-colorings in a graph. 
Discrete Applied Mathematics 160, pp.~2199--2207 (2012)

\bibitem{ItoKaminskiOnoSuzukiUeharaYamanaka2014}
Ito, T., Kami\'nski, M., Ono, H., Suzuki, A., Uehara, R., Yamanaka, K.:
On the parameterized complexity for token jumping on graphs.
Proc.~of TAMC 2014, 
LNCS 8402, 
pp.~341--351 (2014)

\bibitem{IKOZ_isaac}
Ito, T., Kawamura, K., Ono, H., Zhou, X.: 
Reconfiguration of list $L(2,1)$-labelings in a graph. 
To appear in 
Theoretical Computer Science. 
{\tt DOI: 10.1016/j.tcs.2014.04.011}
 
\bibitem{IKZ11}
Ito, T., Kawamura, K., Zhou, X.: 
An improved sufficient condition for reconfiguration of list edge-colorings in a tree. 
IEICE Trans.~on Information and Systems E95-D, pp.~737--745 (2012) 

\bibitem{KMP11}
Kami\'nski, M., Medvedev, P., Milani${\rm \check{c}}$, M.: 
Shortest paths between shortest paths.
Theoretical Computer Science 
412, pp.~5205--5210 (2011)

\bibitem{KaminskiMedvedevMilanic2012}
Kami\'nski, M., Medvedev, P., Milani${\rm \check{c}}$, M.: 
Complexity of independent set reconfigurability problems.
Theoretical Computer Science 
439, pp.~9--15 (2012)

\bibitem{MTY11}
Makino, K., Tamaki, S., Yamamoto, M.:
An exact algorithm for the Boolean connectivity problem for $k$-CNF.
Theoretical Computer Science 
412, pp.~4613--4618 (2011)

\bibitem{MNRSS13}
Mouawad, A.E., Nishimura, N., Raman, V., Simjour, N., Suzuki, A.: 
On the parameterized complexity of reconfiguration problems.
Proc.~of IPEC 2013, 
LNCS 8246, 
pp.~281--294 (2013)

\bibitem{MNRW14}
Mouawad, A.E., Nishimura, N., Raman, V., Wrochna, M.: 
Reconfiguration over tree decompositions.
{\tt arXiv:1405.2447}

\bibitem{Wro14}
Wrochna, M.:
Reconfiguration in bounded bandwidth and treedepth.
{\tt  arXiv:1405.0847} (2014)
\end{thebibliography}

\end{document}